\documentclass[a4paper,12pt]{amsart}
\usepackage{amsfonts}
\usepackage{amsthm}
\usepackage{amssymb}
\usepackage{amstext}
\usepackage{amsmath}
\usepackage{wasysym}
\usepackage{tikz}

\usepackage{graphics}
\usepackage{epsfig}
\usepackage{amscd}
\usepackage[all]{xy}
\usepackage{latexsym}
\usepackage{graphicx}
\usepackage{multirow}
\usepackage{geometry}

\usepackage{fancyhdr}
\usepackage{lipsum}
\usepackage{indentfirst}
\usepackage{graphicx}
\usepackage{textcomp}
\usepackage{faktor}
\usepackage{newlfont}
\usepackage{latexsym}
\usepackage{mathtools}

\DeclarePairedDelimiter\floor{\lfloor}{\rfloor}
\DeclareMathOperator\Sinh{sinh}

\newcommand{\cV}{\mathcal{V}}

\newcommand{\A}{\mathcal{A}}
\newcommand{\Z}{\mathbb{Z}}

\newcommand{\C}{\mathbb{C}}
\newcommand{\cE}{\mathcal{E}}
\newcommand{\br}[1]{\left( #1 \right) }
\newcommand{\bs}[1]{\left\{ #1 \right\} }
\newcommand{\ccor}[1]{\left\langle #1 \right\rangle^\circ }
\newcommand{\cor}[1]{\left\langle #1 \right\rangle }
\newcommand{\deeta}{\delta_{\eta,0}}

\newcommand{\mathpic}[1]{\ensuremath{\vcenter{\hbox{\begin{tikzpicture} #1 \end{tikzpicture}}}}}
\newcommand{\res}{\mathop{\mathrm{Res}}}

\newcommand{\lb}{\left (}
\newcommand{\rb}{\right )}

\theoremstyle{plain}                    
\newtheorem{teo}{Theorem}[section]      
\newtheorem{prop}[teo]{Proposition}    
\newtheorem{lem}[teo]{Lemma}            
\theoremstyle{definition}               
\newtheorem{defin}{Definition}        
\theoremstyle{remark}                   
\newtheorem{rmk}{Remark}           
\newtheorem{notat}{Notation}     

\newcommand\mS{\mathcal{S}}
\newcommand\phm[2]{\lb #1 \rb_{#2}}

\usetikzlibrary{calc}

\newcommand\defaultcolor[0]{black}

\newcommand\youngTableauCell[4]{ 
  \draw[thick] (#1,#2) rectangle +(#4,#4);
  \node[\defaultcolor] at ($(#1,#2) + 0.5*(#4,#4)$) {#3};
}

\newenvironment{youngTableau}[1][]{
  \tikzset {
    cell size/.store in=\cellSize,
    begin x/.store in=\beginX,
    begin y/.store in=\beginY,
  }
  \def\advanceX{\pgfmathsetmacro{\X}{\X + \cellSize + 0.1}}
  \def\advanceY{\pgfmathsetmacro{\Y}{\Y - \cellSize - 0.1}}
  \def\resetX{\pgfmathsetmacro{\X}{\beginX}}
  \def\resetY{\pgfmathsetmacro{\Y}{\beginY}}

  \newcommand\cell[1]{ 
    \youngTableauCell{\X}{\Y}{##1}{\cellSize}
    \advanceX
  }

  \newcommand\RCell[1]{
    \begin{scope}[red]
      \cell{##1}
    \end{scope}
    \advanceX 
  }

  \newcommand\BCell[1]{
    \begin{scope}[blue]
      \cell{##1}
    \end{scope}
    \advanceX 
  }

  \newcommand\WCell[1]{
    \begin{scope}[white]
      \cell{##1}
    \end{scope}
    \advanceX 
  }

  \def\break{
    \resetX \advanceY
  }

  \tikzset{
    cell size=1,
    begin x=0,
    begin y =0,
    #1
  }
  \resetX \resetY
}{}

\numberwithin{equation}{section}

\begin{document}

\title[Polynomiality of orbifold Hurwitz numbers]
      {Polynomiality of orbifold Hurwitz numbers, spectral curve, and a new proof of the Johnson-Pandharipande-Tseng formula}

\author[P.~Dunin-Barkowski]{P.~Dunin-Barkowski}
\address{P.~D.-B.: Korteweg-de~Vries Institute for Mathematics, University of Amsterdam, P.~O.~Box 94248, 1090 GE Amsterdam, The Netherlands; ITEP, Moscow, Russia; and Laboratory of Mathematical Physics, National Research University Higher School of Economics, Moscow, Russia}
\email{P.Dunin-Barkovskiy@uva.nl}


\author[D.~Lewanski]{D.~Lewanski}
\address{D.~L.: Korteweg-de Vries Institute for Mathematics, University of Amsterdam, Postbus 94248, 1090 GE Amsterdam, The Netherlands}
\email{D.Lewanski@uva.nl}

\author[A.~Popolitov]{A.~Popolitov}
\address{A.~P.: Korteweg-de Vries Institute for Mathematics, University of Amsterdam, Postbus 94248, 1090 GE Amsterdam, The Netherlands and ITEP, Moscow, Russia}
\email{A.Popolitov@uva.nl}

\author[S.~Shadrin]{S.~Shadrin}
\address{S.~S.: Korteweg-de Vries Institute for Mathematics, University of Amsterdam, Postbus 94248, 1090 GE Amsterdam, The Netherlands}
\email{S.Shadrin@uva.nl}

\begin{abstract}
In this paper we present an example of a derivation of an ELSV-type formula using the methods of topological recursion. Namely, for orbifold Hurwitz numbers we give a new proof of the spectral curve topological recursion, in the sense of Chekhov, Eynard, and Orantin, where the main new step compared to the existing proofs is a direct combinatorial proof of their quasi-polynomiality. Spectral curve topological recursion leads to a formula for the orbifold Hurwitz numbers in terms of the intersection theory of the moduli space of curves, which, in this case, appears to coincide with a special case of the Johnson-Pandharipande-Tseng formula.
\end{abstract}

\maketitle

\tableofcontents

\section{Introduction}

\subsection{Main goal}
The main goal of this paper is to present a new important application of the procedure that allows to relate in a uniform way a class of combinatorial problems to the intersection theory of the moduli space of curves. First, let us describe this procedure. The logic behind it is the following one:
\begin{itemize}
	\item[--] We start with a combinatorial problem that depends in a natural way on a genus parameter $g\geq 0$ and a vector $\vec{\mu}\in\mathbb{Z}^{n}_{>0}$. 
	\item[--] We consider the generating functions that solve this combinatorial problem. Quite often we can prove that they can be considered as an expansion of certain symmetric differentials $\omega_{g,n}$ that solve the matrix model topological recursion~\cite{EO,Ey} for a particular spectral curve data. 
	\item[--] Under some mild assumptions, the expansion of the symmetric differentials obtained via the topological recursion can be represented (up to some constants) as
	$$
	\sum_{l(\vec{\mu})=n} \sum_{a_1,\dots,a_n=1}^r
	\left[ \int_{\overline{\mathcal{M}}_{g,n}} \frac{S(a_1,\dots,a_n)}{\prod_{j=1}^n (1-\psi_j \frac{d}{dx_j})} \right] \prod_{j=1}^n \xi_{a_j}(x_j).
	$$
	Here $r$ is the number of branching points on the spectral curve, $S(a_1,\dots,a_n)$ is a certain explicitly described tautological class on the moduli space of curves, and $\xi_a(x)$ are some auxiliary functions, $a=1,\dots,r$, also explicitly described~\cite{DOSS12,EYNARD}. 
	\item[--] This way we solve the original combinatorial problem in terms of the intersection numbers of the tautological classes on the moduli space of curves, and the formula that we get is of ELSV-type~\cite{ELSV}. 
\end{itemize}

The first instance of this way to derive an ELSV-type formula was presented in~\cite{DKOSS}, where this leads to a new proof of the original ELSV formula for ordinary Hurwitz numbers.

In this paper we perform this whole procedure for the so-called orbifold Hurwitz numbers~\cite{JPT,J09,DLN,BHLM}. The orbifold Hurwitz numbers are a special case of double Hurwitz numbers~\cite{Joh10}, where the ramification indices in one special fiber are given by an arbitrary partition $\mu$, and in the other special fiber they are all equal to $r$. The intersection formula that we obtain via this procedure was previously derived by Johnson, Pandharipande, and Tseng~\cite{JPT}, and this way we get a new proof of it. 

\subsection{The known facts about orbifold Hurwitz numbers}

Let us collect here the known facts about orbifold Hurwitz numbers so that we can summarize all relevant previous papers about them.
\begin{enumerate}
\item[Fact 1:] \emph{(JPT Formula)} The orbifold Hurwitz numbers are given by the intersection numbers on the moduli space of curves via the Johnson-Pandharipande-Tseng formula.
\item[Fact 2:] \emph{(Quasi-Polynomiality)} The orbifold Hurwitz numbers can be represented, up to a particular combinatorial factor, as the values of a polynomial in $n$ variables $\mu_1,\dots,\mu_n$, whose coefficients depend only on $\vec{\mu}\mod r$.
\item[Fact 3:] \emph{(Cut-and-Join)} The orbifold Hurwitz numbers satisfy a simple recursion with a clear topological meaning, which is called the cut-and-join equation~\cite{GJV}.
\item[Fact 4:] \emph{(Topological Recursion)} The $n$-point generating functions of orbifold Hurwitz numbers can be represented as expansions of the correlation differentials obtained via the Chekhov-Eynard-Orantin topological recursion procedure. 
\end{enumerate}

Let us explain what was known before. First of all, we have the Johnson-Pandharipande-Tseng result itself \cite{JPT}:
\begin{align*}
& \text{(Definition)} \Rightarrow \text{(JPT formula)} \\
\intertext{The main results of~\cite{BHLM} and~\cite{DLN} can be described as follows:}
& \text{(JPT formula)} \Rightarrow \text{(Quasi-Polynomiality)}
\\
& \text{(Quasi-Polynomiality) AND (Cut-and-Join)} \Rightarrow 
\text{(Topological Recursion)}.
\end{align*}
Here the first implication is obvious; though, until now, there was no other proof of quasi-polynomiality than its derivation from the structure of the Johnson-Pandharipande-Tseng formula. So, we see that the JPT formula is used in a very weak way in these papers; only its general structure appears to be relevant.

In~\cite{LPSZ} the full power of the JPT formula is employed; as a result it is proved there that 
$$
\text{(JPT formula)} \Leftrightarrow \text{(Topological Recursion)}
$$

In the present paper, we first give a direct proof of the quasi-polynomiality just from the definition of orbifold Hurwitz numbers. This allows us to use the results of~\cite{BHLM,DLN} in order to prove the topological recursion. This allows us to use the result of~\cite{LPSZ} in order to prove, in a new way, the Johnson-Pandharipande-Tseng formula. So, the structure of this paper can be summarized as follows:
\begin{align*}
& \text{(Definition)}\; \mathop{\Longrightarrow}^{\text{[this paper]}}\; \text{(Quasi-Polynomiality)}
\\
& \text{(Quasi-Polynomiality) AND (Cut-and-Join)}\; \mathop{\Longrightarrow}^{\text{following \cite{BHLM,DLN}}}\;
\text{(Topological Recursion)}
\\
& \text{(Topological Recursion)} \; \mathop{\Longrightarrow}^{\text{using \cite{LPSZ}}}\; \text{(JPT formula)} 
\end{align*}
The first step here is original and it is the main technical result of the present paper; in the second step we follow~\cite{BHLM,DLN}, though we try to focus more on the main structure of the formulas that represent the abstract loop equations rather than on explicit computations; in the third step we just use the results of~\cite{LPSZ}.

\subsection{Organization of the paper} In Section 2 we introduce our basic technical tool --- the semi-infinite wedge formalism. In Section 3 we develop further this formalism, in particular, we use it to define the orbifold Hurwitz numbers, and we represent them, in particular, using the so-called $\mathcal{A}$-operators. 
In Section 4 we analyze further the formula for orbifold Hurwitz numbers in terms of $\mathcal{A}$-operators in order to prove their quasi-polynomiality. In Section 5 we recall the basic setup of the topological recursion. In Section 6 we show how one can use the quasi-polynomiality and the cut-and-join equation for orbifold Hurwitz numbers in order to prove the topological recursion. In Section 7 we use the result of~\cite{LPSZ} to prove the Johnson-Pandharipande-Tseng formula in a new way.

Throughout this paper we fix integer $r\geq 1$.

\subsection*{Acknowledgements}

P.~D.-B., D.~L., A.~P., and S.~S. were supported by the Netherlands Organization for Scientific Research (NWO). P.~D.-B. and A.~P. were also partially supported by the Russian President's Grant of Support for the Scientific Schools NSh-3349.2012.2; P.~D.-B. was partially supported by RFBR grants 13-02-00478 and 14-01-31395-mol\_a, RFBR-India grant 14-01-92691-Ind\_a and RFBR-Turkey grant 13-02-91371-St\_a; A.~P. was partially supported by RFBR grants 13-02-00457 and 14-01-31492-mol\_a. P.~D.-B. was also partially supported by the Government of the Russian Federation within the framework of a subsidy granted to the National Research University Higher School of Economics for the implementation of the Global Competitiveness Program.



\section{Semi-infinite wedge formalism}\label{WEDGE}
In this section we introduce the semi-infinite wedge formalism. This allows us in the Section \ref{OPERATORS} to express $r$-orbifold Hurwitz numbers in terms of vacuum expectation of operators acting on the semi-infinite wedge space.
For a more complete introduction see e.g.~\cite{MiwaJimboDate,OP06,Joh10}.

Let $V$ be an infinite dimensional vector space with a basis labeled by half-integers. Denote the basis vector labeled by $m/2$ by $\underline{m/2}$, so $V = \bigoplus_{i \in \Z + \frac{1}{2}} \underline{i}$.

\begin{defin}\label{DEFSEMIINF}
The semi-infinite wedge space $\bigwedge^{\frac{\infty}{2}}(V) = \mathcal{V}$ is defined to be the span of all of the semi-infinite wedge products of the form
\begin{equation}\label{wedgeProduct}
\underline{i_1} \wedge \underline{i_2} \wedge \cdots
\end{equation}
for any decreasing sequence of half-integers $(i_k)$ such that there is an integer $c$ with $i_k + k - \frac{1}{2} = c$ for $k$ sufficiently large. The constant $c$ is called the \textit{charge}. We give $\mathcal{V}$ an inner product $(\cdot,\cdot)$ declaring its basis elements to be orthonormal.
\end{defin}

\begin{rmk}
By Definition \ref{DEFSEMIINF} the charge-zero subspace $\mathcal{V}_0$ of $\mathcal{V}$ is spanned by semi-infinite wedge products of the form 
$$
\underline{\lambda_1 - \frac{1}{2}} \wedge \underline{\lambda_2 - \frac{3}{2}} \wedge \cdots
$$
for some integer partition~$\lambda$. Hence we can identify integer partitions with the basis of this space:
\begin{equation*}
\mathcal{V}_0 = \bigoplus_{n \in \mathbb{N} } \bigoplus_{\lambda\,  \vdash\, n} v_{\lambda}
\end{equation*}
\end{rmk}

The empty partition $\emptyset$ plays a special role.
We call $$v_{\emptyset} = \underline{-\frac{1}{2}} \wedge \underline{-\frac{3}{2}} \wedge \cdots$$ the vacuum vector and we denote it by $|0\rangle$. Similarly we call the covacuum vector its dual with respect to the scalar product $(\cdot,\cdot)$ and we denote it by $\langle0|$.

\begin{defin}
 The \emph{vacuum expectation value} or \emph{disconnected correlator} $\cor{\mathcal{P}}^{\bullet}$ of an operator~$\mathcal{P}$ acting on $\mathcal{V}_0$: 
is defined to be:
\begin{equation*}
\cor{\mathcal{P}}^{\bullet} : = (|0\rangle , \mathcal{P} |0\rangle) =: \langle 0 |\mathcal{P}|0 \rangle
\end{equation*}
We also define 
\begin{equation}
\zeta(z)=e^{z/2} - e^{-z/2} = 2 \Sinh(z/2)
\end{equation}
\end{defin}

\begin{defin} This is the list of operators we will use:
\begin{enumerate}
\item[i)] For $k$ half-integer the operator
$\psi_k \colon (\underline{i_1} \wedge \underline{i_2} \wedge \cdots) \ \mapsto \ (\underline{k} \wedge \underline{i_1} \wedge \underline{i_2} \wedge \cdots)$
increases the charge by $1$. Its adjoint operator $\psi_k^*$ with respect to~$(\cdot,\cdot)$
 decreases the charge by $1$.
 \item[ii)] The normally ordered products of $\psi$-operators
\begin{equation}
E_{i,j} := \begin{cases}\psi_i \psi_j^*, & \text{ if } j > 0 \\
-\psi_j^* \psi_i & \text{ if } j < 0\ .\end{cases} 
\end{equation}
preserve the charge and hence can be restricted to $\cV_0$ with the following action. For $i\neq j$ $E_{i,j}$ checks if $v_\lambda$ contains $\underline{j}$ as a wedge factor and if so replaces it by $\underline{i}$. Otherwise it yields~$0$. In the case $i=j > 0$, we have $E_{i,j}(v_\lambda) = v_\lambda$ if $v_\lambda$ contains $\underline{j}$ and $0$ if it does not; in the case  $i=j < 0$, we have $E_{i,j}(v_\lambda) = - v_\lambda$ if $v_\lambda$ does not contain $\underline{j}$ and $0$ if it does.
 \item[iii)] The diagonal operators are assembled into the operators
\begin{equation}
\mathcal{F}_n := \sum_{k\in\Z+\frac12} \frac{k^n}{n!} E_{k,k} 
\end{equation}
We will be particularly interested in $\mathcal{F}_2$. The operator $\mathcal{F}_0$ is called \emph{charge operator}, while $\mathcal{F}_1$ is called \emph{energy operator}. Note that $\mathcal{F}_0$ identically vanishes on $\cV_0$, while $\mathcal{F}_1$ has the basis vectors $v_{\lambda}$ as its eigenvectors, with eigenvalues being $|v_\lambda|$ (we refer to $|v_\lambda|$ as the \emph{energy} of basis vector $v_\lambda$). We also say that operator $\mathcal{P}$, defined on $\cV_0$, is an operator of energy $c\in\mathbb{Z}$ if $-[\mathcal{F}_1,\mathcal{P}]$ is proportional to $\mathcal{P}$ with $c$ being the coefficient of proportionality, i.e. if
\begin{equation}
-[\mathcal{F}_1,\mathcal{P}] = c\, \mathcal{P}
\end{equation}
In other words, if $\mathcal{P}$ is an operator of energy $c$, then it maps a basis element of energy $k$ into a combination of basis elements that all have energies $k-c$.

It will be important to us that operators with positive energy annihilate the vacuum while negative energy operators are annihilated by the covacuum, explicitly: let $\mathcal{M}$ be any operator, let $\mathcal{P}$ have positive energy and $\mathcal{N}$ have negative energy, then $\langle \mathcal{MP}\rangle^{\bullet} = 0$ and $\langle \mathcal{NM}\rangle^{\bullet} = 0$. The operator $E_{i,j}$ has energy $j-i$, hence all the $\mathcal{F}_n$'s have zero energy.
\item[iv)] For $n$ any integer and $z$ a formal variable one has the energy $n$ operators:
\begin{equation}
\mathcal{E}_n(z) = \sum_{k \in \Z + 1/2} e^{z(k - \frac{n}{2})} E_{k-n,k} + \frac{\delta_{n,0}}{\zeta(z)}  .
\end{equation}
\item[v)] For $n$ any nonzero integer one has the energy $n$ operators:
\begin{equation}
\alpha_n = \mathcal{E}_n(0) = \sum_{k \in \Z + 1/2} E_{k-n,k}
\end{equation}
\end{enumerate}
\end{defin}

The commutation formula for $\mathcal{E}$ operators is:
\begin{equation}\label{COMME}
  \left[\cE_a(z),\cE_b(w)\right] =
\zeta\left(\det  \left[
\begin{smallmatrix}
  a & z \\
b & w
\end{smallmatrix}\right]\right)
\,
\cE_{a+b}(z+w)
\end{equation}

Note that:
\begin{equation}
\cE_k(z)\big|0\big\rangle = \begin{cases} 
\dfrac{1}{\zeta(z)}\big|0\big\rangle, & \text{ if } k = 0 \\
0 & \text{ if }  k > 0\ .
\end{cases} 
\end{equation}



\section{$\mathcal{A}$ operators}\label{OPERATORS} 
Let $r$ be a positive natural number. The $r$-orbifold Hurwitz numbers~$h^{\bullet,[r]}_{g,\mu}$ enumerate ramified coverings of the 2-sphere by a possibly disconnected genus $g$ surface, where the ramification profile over infinity is given by a partition $\mu=(\mu_1,\dots,\mu_{l(\mu)})$ and the ramification profile over zero is $(r,\dots,r)$, there are simple ramifications over $$b := 2g - 2 + l(\mu) + \sum_{i=1}^{l(\mu)} \frac{\mu_i}{r} $$ fixed points, and there are no further ramifications. Clearly $r$ should divide the degree $d=|\mu|$ of the covering.

\begin{defin}
The genus-generating function of disconnected $r$-orbifold numbers is the following formal power series:
\begin{equation}\label{GENGEN}
\quad H^{\bullet,[r]}(\vec{\mu},u) = \sum_{g \geq 0} h^{\bullet,[r]}_{g, \vec{\mu}} \frac{u^b}{b!} 
\end{equation}
\end{defin}

The disconnected $r$-orbifold Hurwitz numbers can be expressed as vacuum expectation in the following way (see \cite{OP06, Joh10,Ok00}) : 
\begin{equation}\label{HCOMEF}
H^{\bullet,[r]}(\vec{\mu},u) = \sum_{g \geq 0} \bigg{\langle} e^{\frac{\alpha_r}{r} }  \mathcal{F}_2^b \prod_{i=1}^{l(\mu)} \frac{\alpha_{-\mu_i}}{\mu_i} \bigg{\rangle}^{\bullet} \frac{u^b}{b!} = 
\bigg{\langle} e^{\frac{\alpha_r}{r} }  e^{u\mathcal{F}_2} \prod_{i=1}^{l(\mu)} \frac{\alpha_{-\mu_i}}{\mu_i} \bigg{\rangle}^{\bullet}
\end{equation}
We want to express the vacuum expectation in a more convenient way using the so called $\mathcal{A}$ operators introduced in \cite{OP06}. We need the notations:
\begin{notat}
Recall the {\em Pochhammer symbol}:
\begin{equation*} 
(x+1)_n=\frac{(x+n)!}{x!}=\left\{\begin{array}{ll} (x+1)(x+2)\cdots(x+n) &n\geq 0 \\
(x(x-1)\cdots(x+n+1))^{-1} &n\leq 0\end{array} \right. .
\end{equation*}
From the definition, $(x+1)_n$ vanishes for $-n\leq x\leq -1$ an integer, and $1/(x+1)_n$ vanishes for $0\leq x \leq -(n+1)$ an integer.
Let $$\mathcal{S}(z) = \zeta(z)/z = \frac{\Sinh(z/2)}{z/2}$$ Moreover we split rational numbers into integer and fractional parts as follows: for $x\in\mathbb{Q}$ we have
\begin{equation}
x=\floor*{x} + \langle  x \rangle,
\end{equation}
where $\floor*{x}\in \mathbb{Z}$ and $0\leq \langle  x \rangle < 1$.
\end{notat}

\begin{defin}
The following operators will play a central role in the paper:
\begin{equation} 
\label{eq:definition-of-A-operators}
\mathcal{A}^{[r]}_{\eta}(z,u)  =r^{-\eta/r}(\mathcal{S}{(ruz)} )^{\frac{z - \eta}{r}}
\sum_{k \in \mathbb{Z}} \frac{(\mathcal{S}{(ruz)} )^k {z}^k}{(\frac{z - \eta}{r} + 1)_k }   \mathcal{E}_{kr - \eta }(uz)   
\end{equation}
\end{defin}

Define their coefficients in $z$ by $\mathcal{A}^{[r]}_{\eta}(z,u) = \sum_{k \in \mathbb{Z}} \mathcal{A}^{[r],(k)}_{\eta} z^k$.

\begin{rmk}
Our $\mathcal{A}$-operators are at the same time a specialization of Johnson's $\mathcal{A}$-operators in \cite{J09} (which we will denote by ${}_J\mathcal{A}$), and a generalization of Okounkov-Pandharipande ones in \cite{OP06}. Indeed, 
we will specialize Johnson's formulas and results in~\cite{J09} 
using the following assumptions throughout:
\begin{equation}\label{SPEC}
 K = \{e\} \qquad \qquad \qquad R = \mathbb{Z} /r\mathbb{Z}
\end{equation}
This implies that every irreducible representation of $K$ is identically one. With these conditions, Equation (5.5) in~\cite{J09} gives:

$${}_J \mathcal{A}^1_{\frac{a}{r}}(z,u) = \frac{z r^{a/r}}{z + a} \mathcal{S}(ruz)^{\frac{z+a}{r}} 
\sum_{k \in \mathbb{Z}} \frac{(\mathcal{S}(ruz))^k z^k}{\left( \frac{z + a}{r}+ 1 \right)_k} \mathcal{E}_{kr + a}(uz)$$
The two operators agree in the sense that, for $\mu$ positive integers:
\begin{equation*}\label{CFRAOP}
 {}_J\mathcal{A}^1_{1 - \langle \frac{\mu}{r} \rangle}(\mu,u) = \mathcal{A}^{[r]}_{r \langle \frac{\mu}{r} \rangle}(\mu,u) = r^{-\langle \frac{\mu}{r} \rangle}(\mathcal{S}{(ru\mu)} )^{\floor*{\frac{\mu}{r}}}
\sum_{k \in \mathbb{Z}} \frac{(\mathcal{S}{(ru\mu)} )^k {\mu}^k}{(\floor*{\frac{\mu}{r}} + 1)_k }   \mathcal{E}_{kr - r \langle \frac{\mu}{r} \rangle}(u\mu)  
 \end{equation*}

 Johnson defines his semi-infinite wedge space to be a tensor product between
usual semi-infinite wedge space and group $K$. With $K$ specialized to trivial group, however,
his definition reduces to the ordinary semi-infinite wedge space.
\end{rmk}

\begin{prop}\label{HCOMEA} The generating function for disconnected orbifold
Hurwitz can be expressed in terms of the $\mathcal{A}$ operators by:
\begin{equation} \label{eq:formula-Hurwitz-A-operators}
H^{\bullet,[r]}(\vec{\mu},u) = r^{\sum_{i=1}^{l(\vec{\mu})}  \langle \frac{\mu_i}{r}\rangle}
 \prod_{i=1}^{l(\vec{\mu})}  \frac{u^{\frac{\mu_i}{r}}\mu_i^{\floor*{\frac{\mu_i}{r}} - 1}}{\floor*{\frac{\mu_i}{r}} !}
\bigg{\langle}  \prod_{i=1}^{l(\vec{\mu})} \mathcal{A}^{[r]}_{r\langle \frac{\mu_i}{r} \rangle}(\mu_i,u) \bigg{\rangle}^{\bullet}
\end{equation}
\end{prop}

\begin{proof}
Both the operators $\alpha_r$ and $\mathcal{F}_2$ annihilate the vacuum, hence we can conjugate each operator 
$\alpha_{-\mu_i}$ in \eqref{HCOMEF} by their exponent getting:
\begin{equation}\label{eq:first-operator-formula}
H^{\bullet,[r]}(\vec{\mu},u) =
\frac{1}{\prod_{i=1}^{l(\vec\mu)} \mu_i}
\bigg{\langle}  \prod_{i=1}^{l(\vec\mu)} e^{\frac{\alpha_r}{r} }  e^{u\mathcal{F}_2} \alpha_{-\mu_i} e^{-u\mathcal{F}_2} e^{-\frac{\alpha_r}{r}}\bigg{\rangle}^{\bullet}
\end{equation}
We recall Equation (2.14) in \cite{OP06}: 
$$ e^{u\mathcal{F}_2} \alpha_{-\mu} e^{-u\mathcal{F}_2} = \mathcal{E}_{-\mu} (u\mu)$$
Note that the energy is preserved to be $-\mu$. Commutator rule \eqref{COMME} gives:
$$ [\alpha_r,\mathcal{E}_{-\mu}(u\mu)] = \zeta(ru\mu)\mathcal{E}_{r-\mu}(u\mu)$$
We expand the last conjugation in nested commutators of the form above obtaining:
$$ e^{\frac{\alpha_r}{r} }  \mathcal{E}_{-\mu}(u\mu) e^{-\frac{\alpha_r}{r}} = 
\sum_{k \geq 0} \left (\frac{\zeta{(ru\mu)}}{r} \right )^k \frac{1}{k!} \mathcal{E}_{kr - \mu}(u\mu) 
=\sum_{k\geq 0} \frac{u^k\mu^k(\mathcal{S}(ru\mu))^k}{k!}\mathcal{E}_{kr-\mu}(u\mu)$$
Rescaling by $ k - \floor*{ \frac{\mu}{r}} \mapsto k$ and using the vanishing properties of the Pochhammer symbol, we can rewrite the last expression as
$$ \frac{(u\mu)^{\floor*{ \frac{\mu}{r}} } }{\floor*{ \frac{\mu}{r}}!} (\mathcal{S}(ru\mu))^{\floor*{ \frac{\mu}{r}} } \sum_{k \in \mathbb{Z}} \frac{u^k(\mathcal{S}(ru\mu))^k \mu^k}{(\floor*{ \frac{\mu}{r} } +1 )_k}   \mathcal{E}_{kr - \langle \frac{\mu}{r} \rangle r }(u\mu) $$
To match the powers of $u$ we conjugate by the exponent of the energy operator $u^{\mathcal{F}_1/r}$. Since $\mathcal{F}_1$ and its adjoint fix the vacuum, this does not affect operator expectations of products of the $\mathcal{A}$-operators. Since $\mathcal{E}_j$ has energy $j$, the conjugation removes $u^k$ from inside the sum and produces a factor of $u^{\langle \frac{\mu}{r} \rangle }$ outside. Thus we see that the vacuum expectation of the operators in~\eqref{eq:first-operator-formula} can be replaced by the vacuum expectation of the product of 
$$ \frac{u^{\frac{\mu_i}{r}}\mu_i^{\floor*{ \frac{\mu_i}{r}} } }{\floor*{ \frac{\mu_i}{r}}!} (\mathcal{S}(ru\mu_i))^{\floor*{ \frac{\mu_i}{r}} } \sum_{k \in \mathbb{Z}} \frac{(\mathcal{S}(ru\mu_i))^k \mu_i^k}{(\floor*{ \frac{\mu_i}{r} } +1 )_k}   \mathcal{E}_{kr - \langle \frac{\mu_i}{r} \rangle r }(u\mu_i) $$
for $i=1,\dots,l(\vec\mu)$. Then, using Equation~\eqref{eq:definition-of-A-operators} we can rewrite the full formula~\eqref{eq:first-operator-formula} as~\eqref{eq:formula-Hurwitz-A-operators}.
\end{proof}

Following \cite{OP06}, we define the doubly infinite series:
$$ \delta(z,-w) = \frac{1}{w}\sum_{k \in \mathbb{Z}} \left ( -\frac{z}{w}\right )^k $$
which is obtained as the difference between the following two expansions:
$$\frac{1}{z+w} = \frac{1}{w} - \frac{z}{w^2} + \frac{z^2}{w^3} - \dots , \qquad \qquad |z| < |w|$$
$$\frac{1}{z+w} = \frac{1}{z} - \frac{w}{z^2} + \frac{w^2}{z^3} - \dots , \qquad \qquad |z| > |w|$$
The series $\delta(z,-w)$ is a formal $\delta$-function at $z+w=0$ in the sense that:
$$(z+w)\delta(z,-w) = 0$$

We recall the formula for commutators of $\mathcal{A}$, that will be fundamental to prove polynomiality. Below, by $\delta_r(\eta)$ we denote the function of an integer argument that equals to $1$ if $ \eta \equiv 0\mod r$ and vanishes otherwise.

\begin{prop}[Particular case of Lemma V.4. of \cite{J09}]
Let $\eta_1,\eta_2$ be integer numbers satisfying $0\leq \eta_1,\eta_2\leq r-1$.  We have:
\begin{equation}\label{COMM}
[ \mathcal{A}^{[r]}_{\eta_1}(z,u), \mathcal{A}^{[r]}_{\eta_2}(w,u) ] = \delta_r(\eta_1 + \eta_2) zw\delta(z,-w)
\end{equation}
or equivalently:
\begin{equation}\label{COMM2}
 [\mathcal{A}_{\eta_1}^{[r],(k)}, \mathcal{A}_{\eta_2}^{[r],(l)}] = \delta_r(\eta_1 + \eta_2) (-1)^{l} \delta_{k+l-1}.
 \end{equation}
\end{prop}

We define $\Omega\subset \mathbb{C}^n$ by
$$
\Omega=\bigg\{(z_1,\dots, z_n)\in \mathbb{C}^n\bigg|\forall k,|z_k|>\sum_{i=1}^{k-1}|z_i|\bigg\}.
$$

Specializing Theorem V.2 of~\cite{J09} with the convention \eqref{SPEC} (see also Section~2.4 in~\cite{DKOSS}) we have the following:

\begin{prop} For any integer numbers $\eta_1,\dots,\eta_n$, $0\leq \eta_1,\dots,\eta_n\leq r-1$,
the Laurent series expansion of 
$$
\bigg{\langle}  \mathcal{A}^{[r]}_{\eta_1}( z_1, u)\cdots \mathcal{A}^{[r]}_{\eta_1}( z_1, u) \bigg{\rangle}^{\bullet}
$$
in $u,z_1,\dots,z_n$ converges to an analytic function for $(z_1,\dots,z_n)\in\Omega$ and sufficiently small $u\not=0$.
\end{prop}

\begin{notat} For brevity in the rest of the paper we denote  $\mathcal{A}^{[r]}_{\eta}(z,u)$ by $\mathcal{A}_{\eta}(z)$.
\end{notat}



\section{Quasi-polynomiality}\label{POLY}
In this section we derive quasi-polynomiality of $r$-orbifold Hurwitz numbers (Theorem \ref{ThPol}). The argument that we use is a suitable generalization of  an argument of \cite{DKOSS}. 

\subsection{Connected vacuum expectations}
\label{INCL-EXCL-FORMULA}

Proposition \ref{HCOMEA} expresses the genus-generating function of disconnected orbifold Hurwitz in terms of vacuum expectation of $\mathcal{A}$-operators.
Our first goal is have a similar expression for connected orbifold Hurwitz numbers 

\begin{defin} We define the connected correlators $	{\langle}  \mathcal{A}_{\eta_1}( z_1)\cdots \mathcal{A}_{\eta_1}( z_1) {\rangle}^{\circ}$ in terms of the disconnected correlators $\langle\cdots\rangle^{\bullet}$ via the inclusion-exclusion formula. 
\end{defin}

The inverse form of the {inclusion-exclusion formula} reads (cf. \cite{DKOSS}):
\begin{equation}\label{INCEXC}
\cor{\A_{\eta_1}(z_1)\dots\A_{\eta_n}(z_n)}^{\bullet}_k
=\sum_{y \in \mathcal{Y}_{n,k}}\prod_{i=1}^{h(y)}
\ccor{\mathcal{A}_{\eta_{c_{i,1}(y)}}(z_{c_{i,1}(y)})\dots \mathcal{A}_{\eta_{c_{i,l_i(y)}(y)}}(z_{c_{i,l_i(y)}(y)})}_{\lambda_i(y)}
\end{equation}
Here $\mathcal{Y}_{n,k}$ is the finite set of $\bs{1,\dots,n}$-Young tableaux $y$ with the following properties:
\begin{enumerate}
	\item The numbers in the rows should be ascending: $c_{i,j}(y)$ is the number in the $i$-th row and $j$-th column, then for any $i$ and for any $j_1 < j_2$ we have $c_{i,j_1}(y) < c_{i,j_2}(y)$. Each row corresponds to an individual connected correlator.
	\item For rows of the same length, just for the first column the numbers should be ascending:  $l_i(y)$ is length of the $i$-th row, then if $l_{i_1}(y)=l_{i_2}(y)$ and $i_1<i_2$, then $c_{i_1,1}(y) < c_{i_2,1}(y)$.
	\item $h(y)$ is the number of rows. Rows are labelled by the vector $\{\lambda_i(y) \in\bs{-1,0,1,\dots} \}_i$ with $\sum_{i=1}^{h(y)}\lambda_i(y) = k$. The vector $\vec{\lambda}$ corresponds to the vector of Euler characteristics of correlators with sign exchanged.
\end{enumerate}

\mathpic{[scale=0.5]
	\node [right] at (0 + 1,0 + 2) {Trivial example};
	\begin{youngTableau}[begin x = 0]
		\WCell{-1} \cell{1}\cell{2}\cell{3}\cell{4}\cell{5} \break
		\WCell{0} \cell{6}\cell{7}\cell{8}\cell{9} \break
		\WCell{1} \cell{10}\cell{11}\cell{12}\cell{13} \break
		\WCell{2} \cell{14}
	\end{youngTableau}
	\node [align=center] at (10 + 1 + 2.5,0 + 2) {More complicated example \\ (allowed disorder marked blue)};
	\begin{youngTableau}[begin x = 10]
		\WCell{-1} \BCell{2}\cell{5}\cell{6}\cell{7}\cell{8} \break
		\WCell{0} \BCell{1}\cell{9}\BCell{12}\cell{13} \break
		\WCell{1} \BCell{4}\cell{10}\BCell{11}\cell{14} \break
		\WCell{2} \BCell{3}
	\end{youngTableau}
	\node [align=center] at (20 + 1 + 2.5,0 + 2) {Incorrect example \\ (errors marked red)};
	\begin{youngTableau}[begin x = 20]
		\WCell{-1} \cell{2}\cell{5}\RCell{7}\RCell{6}\cell{8} \break
		\WCell{0} \RCell{4}\cell{9}\cell{12}\cell{13} \break
		\WCell{1} \RCell{1}\cell{10}\cell{11}\cell{14} \break
		\WCell{2} 
		\cell{3}
	\end{youngTableau}
}

\begin{rmk}
	For $n=1$ we have that connected and disconnected correlators coincide, hence we just write $\cor{\A_{\eta}(z)}$.
\end{rmk}

The connected correlators can be used to express the generating function for connected orbifold Hurwitz numbers:
\begin{equation}
\quad H^{\circ,[r]}(\vec{\mu},u): = \sum_{g \geq 0} h^{\circ,[r]}_{g, \vec{\mu}}.
\frac{u^b}{b!} 
\end{equation}

\begin{prop}\label{CONNGEN}
Generating function for connected orbifold Hurwitz numbers equals:
\begin{equation} \label{CONNGENGEN}
H^{\circ, [r]}(      \vec{\mu},u) = r^{\sum_{i=0}^{l(\vec\mu)} \langle \frac{\mu_i}{r} \rangle}
 \prod_{i=1}^{l(\vec{\mu})}  \frac{u^{\frac{\mu_i}{r}}\mu_i^{\floor*{\frac{\mu_i}{r}} - 1}}{\floor*{\frac{\mu_i}{r}} !}
\bigg{\langle}  \prod_{i=1}^{l(\vec{\mu})} \mathcal{A}_{r\langle \frac{\mu_i}{r} \rangle}(\mu_i) \bigg{\rangle}^{\circ}
\end{equation}
\end{prop}

\begin{proof}
This follows from \eqref{eq:formula-Hurwitz-A-operators} and the observation that
taking $u^b$-coefficient in $H^\circ$ corresponds to 
the coefficient
of $u^{2 g - 2 + l(\vec\mu)}$ in $\langle \prod \mathcal{A} \rangle^{\circ}$.
\end{proof}

\subsection{Unstable terms} In this Section we compute explicitly the coefficients of the connected vacuum expectations that correspond to the orbifold Hurwitz number for $g=0$ and $n=1,2$. 

First, let us introduce some convenient notations.
\begin{notat}
For any operator $\mathcal{P}(u)$ define
\begin{align}
\cor{\mathcal{P}(u)}^{\bullet}_k &:= [u^k]\cor{\mathcal{P}(u)}^{\bullet} & (\mathrm{the\ coefficient\ of\ } u^k\  \mathrm{in}\ \cor{\mathcal{P}(u)}^{\bullet}) \\ \nonumber
\ccor{\mathcal{P}(u)}_k &:= [u^k]\ccor{\mathcal{P}(u)} & (\mathrm{the\ coefficient\ of\ } u^k\  \mathrm{in}\ \ccor{\mathcal{P}(u)})
\end{align}
\end{notat}

\begin{notat} We denote by $\A_{\eta,+}(z)$ the positive power part in $z$ of the $\A_\eta(z)$ operator to be:
	\begin{equation}\label{POSPOW}
	\A_{\eta,+}(z) := \sum_{k \geq 1} \A_{\eta}^{(k)}z^k
	\end{equation}
\end{notat}


The terms that we want to compute are
\begin{equation}
\ccor{\A_{\eta_i}(z_i)}_{-1} \qquad \qquad \text{and} \qquad \qquad \ccor{\A_{\eta_i}(z_i)\A_{\eta_j}(z_j)}_{0}
\end{equation}

\begin{lem}
Let $\eta,\eta_1,\eta_2$ be integer number, $0\leq \eta\leq r-1$. We have:
\begin{align}
\label{1pt}
\ccor{\A_{\eta}(z)}_{-1} & =  \dfrac{\deeta}{z} ,
\\
\label{2pt}
\ccor{\A_{\eta_1}(z_1)\A_{\eta_2}(z_2)}_{0} & = \delta_r(\eta_1 + \eta_2)z_1 \sum_{k \geq 0}  \left( - \frac{z_1}{z_2}\right)^k.
\end{align}
\end{lem}

\begin{proof} In the vacuum expectation of a single operator $\mathcal{A}_\eta(z)$ only zero-energy term can give non-trivial contribution. Since $\mathcal{E}_i$ has energy $i$, we have:
\begin{equation}\label{DEVA}
\cor{\A_{\eta}(z)} = \delta_{\eta,0}\frac{\zeta(ruz)^{z/r}}{(ruz)^{z/r}}\frac{1}{\zeta(uz)} =
   \left[ \dfrac{1}{uz}+\dfrac{z(rz-1)}{24}u+O(u^2) \right] \deeta
\end{equation}
This implies the formula for the genus-zero one-point correlator. The rest of the proof is devoted to the genus-zero two-pointed correlator.

Note that the following formula for the action of $\A_{\eta}(z)$ on covacuum holds
\begin{equation}\label{trikkone}
\langle 0|\A_{\eta}(z) = \dfrac{\deeta}{uz}\langle 0| + \langle 0| \A_{\eta,+}(z),
\end{equation}
which follows directly from Equation~\eqref{eq:definition-of-A-operators} and two
observations:
\begin{itemize}
\item $\mathcal{E}_{kr - \eta }(uz)$ annihilates the covacuum when $kr - \eta < 0$
\item Among the terms that do not annihilate the covacuum, only the term with $\mathcal{E}_{0}(uz)$ is singular in $z$ at $z = 0$
\end{itemize}

Equation~\eqref{INCEXC} implies that 
\begin{align*}
\ccor{\A_{\eta_1}(z_1)\A_{\eta_2}(z_2)}_{0}
 = \cor{\A_{\eta_1}(z_1)\A_{\eta_2}(z_2)}^{\bullet}_0  
 - \cor{\A_{\eta_1}(z_1)}_{-1}\cor{\A_{\eta_2}(z_2)}_{1} \\
 -\cor{\A_{\eta_1}(z_1)}_{1}\cor{\A_{\eta_2}(z_2)}_{-1}
\end{align*}
Applying \eqref{trikkone} to the first term in the right-hand side we get:
\begin{equation*}
\cor{\A_{\eta_1}(z_1)\A_{\eta_2}(z_2)}^{\bullet}_0 = \cor{\A_{\eta_1,+}(z_1)\A_{\eta_2}(z_2)}^{\bullet}_0+\cor{\A_{\eta_1}(z_1)}_{-1}\cor{\A_{\eta_2}(z_2)}_{1}
\end{equation*}
In the same way, we observe that 
\begin{align*}
\cor{\A_{\eta_1,+}(z_1)\A_{\eta_2}(z_2)}^{\bullet}_0 = \cor{[\A_{\eta_1,+}(z_1),\A_{\eta_2}(z_2)]}^{\bullet}_0+\cor{\A_{\eta_2,+}(z_2)\A_{\eta_1,+}(z_1)}^{\bullet}_0
\\ +  \cor{\A_{\eta_1}(z_1)}_{1}\cor{\A_{\eta_2}(z_2)}_{-1}
\end{align*}
Therefore,
\begin{equation*}
\ccor{\A_{\eta_1}(z_1)\A_{\eta_2}(z_2)}_{0} = 
\cor{\A_{\eta_2,+}(z_2)\A_{\eta_1,+}(z_1)}^{\bullet}_0 +  \cor{[\A_{\eta_1,+}(z_1),\A_{\eta_2}(z_2)]}^{\bullet}_0
\end{equation*}
The second term here is equal to the right hand side of Equation~\eqref{2pt} (this follows the commutation rule for coefficients given by Equation~\eqref{COMM2}). In order to complete the proof of the lemma we have to prove that the first term vanishes.

In other words, we consider
\begin{multline} \label{ZEROPOW}
r^{\frac{\eta_1+\eta_2}{r}}\big{\langle} \mathcal{A}_{\eta_2}(z_2)\mathcal{A}_{\eta_1}(z_1) \big{\rangle}^{\bullet} =
(\mathcal{S}{(ruz_2)} )^{\frac{z_2 - \eta_2}{r}} (\mathcal{S}{(ruz_1)} )^{\frac{z_1 - \eta_1}{r}} \times \\
\times \sum_{k,l \in \mathbb{Z}} \frac{(\mathcal{S}{(ruz_2)} )^k {z_2}^k (\mathcal{S}{(ruz_1)} )^l {z_1}^l}{(\frac{z_2 - \eta_2}{r} + 1)_k (\frac{z_1 - \eta_1}{r} + 1)_l }  \big{\langle} \mathcal{E}_{kr - \eta_2 }(uz_2) \mathcal{E}_{lr - \eta_1 }(uz_1)  \big{\rangle}^{\bullet}
\end{multline}
We want to show that the coefficient of $u^0$ in this expression does not 
contain terms of expansion in $z_1,z_2$ that have positive degrees in both variables.
This implies directly that $\cor{\A_{\eta_2,+}(z_2)\A_{\eta_1,+}(z_1)}^{\bullet}_0=0$.

There are two cases: 

\begin{itemize}
\item $k r - \eta_2 = k r - \eta_2 = 0$, which implies
$k = l = \eta_1 = \eta_2 = 0$. In this case the expression \eqref{ZEROPOW} is equal to
\begin{align}
  \label{eq:zero-expr}
  \frac{\mathcal{S}{(ruz_2)}^{\frac{z_2}{r}} \mathcal{S}{(ruz_1)}^{\frac{z_1}{r}}}{\zeta(u z_2) \zeta(u z_1)}
  = \frac{1}{u^2 z_1 z_2} + \frac{1}{24 z_1 z_2} \lb r z_1^3 + r z_2^3 - z_1^2 - z_2^2 \rb + O(u^2),
\end{align}
hence all terms in the coefficient of $u^0$ have negative degree either in $z_1$ or in $z_2$.

\item $k r - \eta_2 \neq 0$ and $l r - \eta_1 \neq 0$, which implies $k r - \eta_2+l r - \eta_1 =0$. In this case all factors are formal power series in $u$, so we can expand all factors in $u$ up to $O(u^1)$.
  The summand with particular $k$ and $l$ in \eqref{ZEROPOW}  is equal to
  \begin{align}
    \frac{z_2^k z_1^l}{(\frac{z_2 - \eta_2}{r} + 1)_k (\frac{z_1 - \eta_1}{r} + 1)_l} + O(u)
  \end{align}
The condition $k r - \eta_2+l r - \eta_1 =0$ is satisfied in one of the two possible cases:
  \begin{itemize}
  \item[--] $\eta_1=\eta_2 = 0$, $k + l = 0$, $k,l\not=0$;
  \item[--] $\eta_1 + \eta_2 = r$, $k + l = 1$.
  \end{itemize}
  In both cases either $k$ or $l$ is non-positive.  Without loss of generality, let's assume that $l \leq 0$ (the other case is symmetric). Then
  \begin{align}
    \frac{z_1^l}{(\frac{z_1 - \eta_1}{r} + 1)_l} = z_1^l
    \lb \frac{z_1 - \eta_1}{r} \cdot \dots \cdot \lb \frac{z_1 - \eta_1}{r} + l + 1 \rb \rb
  \end{align}
  contains no positive powers of $z_1$.
\end{itemize}
\end{proof}

\subsection{Vacuum expectations without unstable terms} In this section we give a formula for the disconnected vacuum expectations, where all unstable terms, that is, $\ccor{\A_{\eta}(z)}_{-1}$ and $\ccor{\A_{\eta_1}(z_1)\A_{\eta_2}(z_2)}_{0}$, are dropped. This is a straightforward generalization of the similar formula in~\cite{DKOSS}, which is based on the following simple recursion rules:

\begin{lem}We can recursively decompose disconnected correlators as follows:
\begin{align}\label{REC1}
& \langle \A_{\eta}(z) \prod_i \A_{\eta_{i}}(z_{i}) \rangle^{\bullet}_k  = 
\langle \A_{\eta}(z)\rangle^{\circ}_{-1} \langle  \prod_i \A_{\eta_{i}}(z_{i}) \rangle^{\bullet}_{k+1} 
+\langle \A_{\eta,+}(z) \prod_i \mathcal{A}_{\eta_i}(z_{i}) \rangle_k^{\bullet}
\\
\label{REC2}
&\langle \A_{\eta,+}(z) \A_{\sigma}(w) \prod_i \A_{\eta_{i}}(z_{i}) \rangle_k^{\bullet} = 
\langle \A_{\eta}(z) \A_{\sigma}(w)\rangle^{\circ}_{0} \langle \prod_i \A_{\eta_{i}}(z_{i}) \rangle_k^{\bullet} 
\\ \notag
 &\phantom{\langle \A_{\eta,+}(z) \A_{\sigma}(w) \prod_i \A_{\eta_{i}}(z_{i}) \rangle_k^{\bullet} = }
 +
\langle \A_{\sigma}(w)  \A_{\eta,+}(z) \prod_i \A_{\eta_{i}}(z_{i}) \rangle_k^{\bullet}
\end{align}
\end{lem}

\begin{proof}
Equation~\eqref{trikkone} and the formula for the one-point correlator~\eqref{1pt} together prove the first equality.
 The second equality follows from the computation of the two-points correlator~\eqref{2pt}.
\end{proof}

This implies the following proposition.

\begin{prop} We have:
\begin{align}\label{STAB}
&
 \cor{\A_{\eta_1,+}(z_1)\dots\A_{\eta_k,+}(z_n)}_k^{\bullet} \\ \notag &
 =\sum_{y \in \mathcal{Y}^{stab}_{n,k}}\prod_{i=1}^{h(y)}\ccor{\A_{\eta_{c_{i,1}(y)}}(z_{c_{i,1}(y)})\dots\A_{\eta_{c_{i,l_i(y)}(y)}}(z_{c_{i,l_i(y)}(y)})}_{\lambda_i(y)}.
\end{align}
where \begin{equation}
\mathcal{Y}^{stab}_{n,k} = \bs{y\in\mathcal{Y}_{n,k}\big|\, l_i(y)=1 \Rightarrow \lambda_i(y)\neq -1,\; l_i(y)=2 \Rightarrow \lambda_i(y)\neq 0}
\end{equation}
In other words, $\cor{\A_{\eta_1,+}(z_1)\cdots\A_{\eta_n,+}(z_n}_k^{\bullet}$ is equal to $\cor{\A_{\eta_1}(z_1)\cdots\A_{\eta_n}(z_n)}_k^{\bullet}$ with all the unstable terms dropped.
\end{prop}

\begin{proof}
	The proof of this proposition is completely analogous to the proof of Proposition~2.21 in~\cite{DKOSS}. It is based on the recursion that expresses 
	$ \langle \A_{\eta_1}(z_1) \cdots \A_{\eta_n}(z_n) \rangle_k^{\bullet}$ in terms of the unstable vacuum expectations and $\A_+$-operators using only Equations~\eqref{REC1} and~\eqref{REC2}. Though the operators here are more general, the recursion rules are still the same, so the same argument can be applied.
\end{proof}

\begin{rmk} Let us give some example to explain the difference between $\mathcal{Y}_{n,k}$ and $\mathcal{Y}^{stab}_{n,k}$. The first example below belongs to $\mathcal{Y}_{3,-1}$ but not to $\mathcal{Y}^{stab}_{3,-1}$, and the second one belongs to $\mathcal{Y}^{stab}_{3,-1}$:
\begin{equation*}
  \mathpic{[scale=0.5]
    \node [align=center] at (1 + 1,0 + 2) {Unstable};
    \begin{youngTableau}[begin x = 0]
      \WCell{0} \cell{1}\cell{2} \break
      \WCell{-1} \cell{3}
    \end{youngTableau}
    \pgftransformshift{\pgfpoint{5 cm}{0}}
    \node [align=center] at (1 + 1,0 + 2) {Stable};
    \begin{youngTableau}[begin x = 0]
      \WCell{-1} \cell{1}\cell{2} \break
      \WCell{0} \cell{3}
    \end{youngTableau}
  }
\end{equation*}
\end{rmk}


\subsection{Polynomiality} In this section we prove the quasi-polynomiality property for orbifold Hurwitz numbers. First, we show that ${\cor{\A_{\eta_1,+}(z_1)\dots\A_{\eta_n,+}(z_n)}_k^{\bullet}}/({z_1 \cdots z_n})$ is a symmetric polynomial in $z_1,\dots,z_n$ (excluding unstable cases of $k=-1,n=1$, and $k=0, n=2$). This implies that $\ccor{\A_{\eta_1}(z_1)\dots\A_{\eta_n}(z_n)}_k/({z_1 \cdots z_n})$ is a symmetric polynomial in $z_1,\dots,z_n$ (again, excluding unstable cases). This, in turn, implies quasi-polynomiality of orbifold Hurwitz numbers.

\begin{prop}
\label{Aplusprop}
The function
\begin{equation}
\displaystyle\dfrac{\cor{\A_{\eta_1,+}(z_1)\dots\A_{\eta_n,+}(z_n)}_k^{\bullet}}{z_1 \cdots z_n}
\end{equation}
is a symmetric polynomial in $z_1,\dots,z_n$ for $(n,k)\not=(1,-1),(2,0)$ .
\end{prop}
\begin{proof}
We follow the proof of Proposition 9 in~\cite{OP06}. We have:
\begin{enumerate}
\item[i)] {Boundedness from below:} $\cor{\A_{\eta_1,+}(z_1)\dots\A_{\eta_n,+}(z_n)}_k^{\bullet}$ has strictly positive powers in all its variables $z_1,\dots,z_n$, as it follows from the definition of $\A_{\eta,+}(z)$ given by Equation~\eqref{POSPOW}. So, we can divide by $\prod_{i=1}^n z_i$, and we still have only non-negative powers of $z_1,\dots,z_n$ in the expansion of the quotient.
\item[ii)] Symmetry  holds because $\A_+$ operators commute which each other, which is a direct consequence of the commutation formula~\eqref{COMM2}.
\item[iii)] {Boundedness from above:} Since it is a symmetric function, it is enough to show that the power of $z_n$ is bounded. From the definition of the $\mathcal{A}$-operator we have the following
\begin{align}\label{eq:action-a-vacuum}
\A_{\eta}(z) |0\rangle=
r^{-\eta/r} \mS (ruz)^{(z-\eta)/r} \sum_{k = 0}^\infty
\frac{\mS(ruz)^{-k} z^{-k}}{\phm{\frac{z - \eta}{r} + 1}{-k}} \cE_{-k r - \eta} (u z) |0\rangle
\end{align}
where we used the change of summation index $k \mapsto -k$ since the operators $\mathcal{E}_i$ with positive $i$ annihilate the vacuum.

Since each factor in each summand in~\eqref{eq:action-a-vacuum} has at most first order pole in $u$, it is sufficient to do the following. We expand  each factor of each summand in~\eqref{eq:action-a-vacuum} in $u$ up to $O(u^{m+1})$, and we show that the degree of $z$ in this expansion is bounded from above. Indeed, in this case, the highest possible power of $z$ in $\cE_{-k r - \eta} (u z)$ is $m$; in $\mS(ruz)^{-k}$ it is again $m$; in $\mS (ruz)^{(z-\eta)/r}$ is it equal to $2m$ (one $m$ comes from argument of $\mS$, while
the other estimates power of $z$ in the binomial coefficient in the expansion of $(1+x)^{(z-\eta)/r}$; finally, the highest possible power of $z$ in 
${z^{-k}}/{\phm{\frac{z - \eta}{r} + 1}{-k}}$ is equal to $0$.
\end{enumerate}
\end{proof}

\begin{prop}
\label{Praconpol}
For $(n,k)\not=(1,-1),(2,0)$, the function
\begin{equation}
\displaystyle\dfrac{\ccor{\A_{\eta_1}(z_1)\dots\A_{\eta_n}(z_n)}_k}{z_1 \cdots z_n}
\end{equation}
is a symmetric polynomial in $z_1,\dots,z_n$.
\end{prop}
\begin{proof}
We follow the proof of Proposition 2.23 in~\cite{DKOSS}. 
We prove the statement by induction on the number of operators in the vacuum expectation. It holds for $n=1$. Suppose it holds for vacuum expectations with any number of operators less than $n$, and we want to prove it for $n$ operators as well.
Let $y'$ be the single-row Young tableau. Consider the partition 
$\mathcal{Y}^{stab}_{n,k} = \{y'\} \cup \left( \mathcal{Y}^{stab}_{n,k} \setminus \{y'\} \right) $.
 Then Equation~\eqref{STAB} implies:
\begin{align}
 &\dfrac{\ccor{\A_{\eta_1}(z_1)\dots\A_{\eta_n}(z_n)}_k}{z_1 \cdots z_n} =
 \dfrac{\cor{\A_{\eta_1,+}(z_1)\dots\A_{\eta_n,+}(z_n)}_k^{\bullet}}{z_1 \cdots z_n}\\ \nonumber
 &-\sum_{y \in \mathcal{Y}^{stab}_{n,k} \setminus \{y'\}}\prod_{i=1}^{h(y)}\dfrac{\ccor{\A_{\eta_{z_{c_{i,1}(y)}}}(z_{c_{i,1}(y)})\dots\A_{\eta_{c_{i,l_i(y)}(y)}}(z_{c_{i,l_i(y)}(y)})}_{\lambda_i(y)}}{z_{c_{i,1}(y)} \cdots z_{c_{i,l_i(y)}(y)}} .
\end{align}
The first term on the right hand side is symmetric polynomial in $z_1,\dots,z_n$ by Proposition~\ref{Aplusprop}, the second term is symmetric polynomial by induction hypothesis.
\end{proof}

Now we are ready to prove the quasi-polynomiality of orbifold Hurwitz numbers.
 \begin{teo} 
 	\label{ThPol}
The orbifold Hurwitz numbers $h^{\circ,[r]}_{g,\mu}$ for $(g,n)\not={(0,1),(0,2)}$ can be expressed as follows:
\begin{equation}\label{eq:orbifold-hurwitz-poly}
h^{\circ,[r]}_{g;\mu} = (2g-2+l(\mu) +|\mu|/r)!\br{\prod_{i=1}^n\dfrac{\mu_i^{\floor*{\frac{\mu_i}{r}}}}{\floor*{\frac{\mu_i}{r}}!}}\, P_{g,n}^{\langle \frac{\vec{\mu}}{r} \rangle}(\mu_1,\dots,\mu_n),
\end{equation}
where $P_{g,n}^{\vec\epsilon}(\mu_1,\dots,\mu_n)$ are some
polynomials in $\mu_1,\dots,\mu_n$, whose coefficients depend on the parameters $\vec\epsilon=(\epsilon_1,\dots,\epsilon_n)$, $0\leq \epsilon_1,\dots,\epsilon_n\leq r-1$
(we have $\epsilon_i=\langle\frac{\mu_i}{r}\rangle$).
\end{teo}
\begin{proof}
This is a direct corollary of Equation~\eqref{CONNGENGEN} and Proposition~\ref{Praconpol}.
\end{proof}



\section{Topological Recursion}\label{EOTR}
In this section we recall the topological recursion of Chekhov, Eynard, and Orantin tailored for our use. For a more detailed introduction we refer to~\cite{EO,Ey}.

\begin{defin} 
A \emph{spectral curve} is a triple $(\Sigma,x,y)$, where 
$\Sigma$ is a Riemann surface (which we assume from now on to be $\C\mathbb{P}^1$) and $x,y\colon\Sigma \to \C$ are meromorphic functions, such that the zeroes of $dx$ are disjoint from the zeroes of $dy$. Moreover the zeros of $dx$ are simple: 

In a neighborhood of a point $\alpha\in\Sigma$ such that $dx(\alpha)=0$ we can define an involution $\tau_\alpha$ that preserves function $x$ (deck transformation).

Furthermore, $\Sigma\times \Sigma$ is equipped with a meromorphic symmetric 2-differential
with a second order pole on the diagonal, which is called the \emph{Bergman kernel}. In the case of $\C\mathbb{P}^1$  the Bergman kernel is unique and in a global coordinate $z$ it reads
$$
B = \frac{dz_1 dz_2}{(z_1 - z_2)^2}.
$$
\end{defin}

\begin{defin} 
	By \emph{topological recursion} we call a recursive procedure that associates to a spectral curve data $(\Sigma,x,y,B)$ a family of symmetric meromorphic differentials, called \emph{correlation differentials} $\omega_{g,n}(z_1,\dots,z_n)$ defined on $\Sigma^n$, $g\geq 0$, $n\geq 1$.
	

The first two correlation differentials are given by explicit formulas:
\begin{equation}\label{eq:initial-values}
\omega_{0,1}(z)=\frac{y\,dx}{x}
\qquad \qquad
\omega_{0,2}(z_1,z_2)=\frac{dz_1 dz_2}{(z_1-z_2)^2} 
\end{equation}
The correlation differentials $\omega_{g,n}$, $2g-2+n>0$, are given by:
 \begin{multline}
\omega_{g,n}(z_1,z_{S})=  
\sum_{\substack{\alpha\in \Sigma \\ dx(\alpha)=0}}\res_{z=\alpha}K(z_1,z) \biggr[\omega_{g-1,n+1}(z,\tau_\alpha(z),z_{S})+
\\
  \mathop{\sum_{\substack{g_1+g_2=g \\ I\sqcup J=S}}'}
\omega_{g_1,|I|+1}(z,z_I)\omega_{g_2,|J|+1}(\tau_\alpha(z),z_J)\biggr],
\end{multline}
where $S=\{2,\ldots,n\}$, and in the second sum we exclude the cases when $(g_1,|I|+1)$ or $(g_2,|J|+1)$ is equal to $(0,1)$. The \emph{recursion kernel} $K$ is defined in the vicinity of each point $\alpha$, $dx(\alpha)=0$ by the formula
\begin{equation}
K(z_1,z):= \frac{\int_z^{\tau_\alpha(z)} \omega_{0,2}(\cdot, z_1)}{2(\omega_{0,1}(\tau_\alpha(z))-\omega_{0,1}(z))}
\end{equation}
In our case ($\Sigma=\C\mathbb{P}^1$, $z$ is a global coordinate), we can use the following formula:
 \begin{equation}
K(z_1,z)
=\frac{x(z)}{2(y(\tau_\alpha(z))-y(z))x'(z)}\left(  \frac{1}{z-z_1}- \frac{1}{\tau_\alpha(z)-z_1}\right)\frac{dz_1}{dz}
\end{equation}
\end{defin}

In the stable range, $2g-2+n>0$, the correlation differentials $\omega_{g,n}$ have poles only at the zeros of $dx$. They can be expressed as the sum of their \emph{principle parts}:
\begin{equation}\label{eq:principle-parts}
\omega_{g,n}(z_1,z_S)=\sum_{\substack{\alpha\in \Sigma \\ dx(\alpha)=0}} [\omega_{g,n}(z_1,z_S)]_\alpha
\end{equation}
where by principal part $[\eta(z_1)]_\alpha$ of a $1$-form $\eta(z_1)$ we mean the projection defined as a version of Cauchy formula, where we use $B$ instead of the Cauchy kernel:
 \begin{equation}
 \left[\eta(z_1)\right]_\alpha := \res_{z = \alpha}\ \eta(z) \int^z_\alpha B(\cdot,z_1).
 \end{equation}
 
In fact, there is an equivalent way to reformulate the topological recursion. We say that the symmetric meromorphic differentials satisfy the topological recursion if they satisfy the property~\eqref{eq:principle-parts} for $2g-2+n>0$, and also solve the abstract loop equations:
\begin{align} \label{eq:linear-loop-equation}
& \omega_{g,n}(z,z_S)+\omega_{g,n}(\tau_\alpha(z),z_S)
 \text{ is holomorphic for } z\to\alpha \\ \label{eq:quadratic-loop-equation}
& \omega_{g-1,n+1}(z,\tau_\alpha(z),z_S)+\sum_{\substack{g_1+g_2=g\\I\sqcup J=S}} \omega_{g_1,1+|I|}(z,z_I)\omega_{g_2,1+|J|}(\tau_\alpha(z),z_J)
\\ \notag &
\text{is holomorphic for } z\to\alpha \text{ with at least double zero at } \alpha.
\end{align}
A proof of that can be found in~\cite{BEO,BS}.

%




\section{The Spectral Curve}
In this section we prove the spectral curve for orbifold Hurwitz numbers using the quasi-polynomiality property
 proved in Section~\ref{POLY} and the cut-and-join equation. It is important to stress that we do not use the Johnson-Pandharipande-Tseng~\cite{JPT} formula in this Section.
 
We consider the $n$-point function  for orbifold Hurwitz numbers for fixed genus $g$:
\begin{equation} \label{orb_gen_func}
H^{\circ,[r]}_{g,n}(x_1, \ldots, x_n) = \sum_{\vec{\mu}: l(\mu)=n} \frac{h^{ \circ,[r]}_{g;\mu}}{b!} ~x_1^{\mu_1} \cdots x_n^{\mu_n}.
\end{equation}
\begin{teo} 
	\label{TOPREC} Consider the correlation differentials $\omega_{g,n}$, $g\geq 0$, $n\geq 1$, for the spectral curve $(\Sigma=\C\mathbb{P}^1,z,y)$, where
\begin{equation}\label{CURVE}
x(z) = z\exp(-z^r) \qquad \text{and} \qquad y(z) = z^r
\end{equation}
in some global coordinate $z$. They have the following analytic expansion near $x_1 = x_2 = \cdots = x_n = 0$:
\begin{equation} \label{eq:orb-gen-form}
\omega_{g,n}(x_1,\dots,x_n) = \frac{\partial}{\partial x_1}\cdots \frac{\partial}{\partial x_n} H^{\circ,[r]}_{g,n}(x_1, \ldots, x_n) dx_1\otimes\cdots\otimes dx_n.
\end{equation}
for all $(g,n)\not=(0,2)$ For $(g,n)=(0,2)$ we have:
\begin{align} \label{eq:0-2-unstable-differentials}
& \omega_{0,2}(x_1,x_2) =\frac{dz(x_1)\otimes dz(x_2)}{(z(x_1)-z(x_2))^2} \\ \notag
	& = \frac{\partial}{\partial x_1}\frac{\partial}{\partial x_2} H^{\circ,[r]}_{0,2}(x_1, x_2) dx_1\otimes dx_2
	+\frac{dx_1\otimes dx_2}{(x_1-x_2)^2}.
\end{align}
\end{teo}

This theorem is proved in \cite{DLN} and \cite{BHLM} using the Johnson-Pandharipande-Tseng formula and the cut-and-join equation for the orbifold Hurwitz numbers. We show below that it is enough to use the quasi-polynomiality property given in Theorem~\ref{ThPol} instead of the Johnson-Pandharipande-Tseng formula. 

Since, except for the first few steps that have to be adjusted, the arguments of~\cite{DLN} and~\cite{BHLM} still work, we refer to these papers for complete computations. Here, after a careful analysis of the consequences of quasi-polynomiality,  we just sketch the main big steps of computation in order to give the reader an idea how the abstract loop equations emerge in this context.

\begin{proof}[Proof of Theorem~\ref{TOPREC}] First of all, we have to check the formulas for $\omega_{0,1}$ and $\omega_{0,2}$. This can be done by direct inspection, see~\cite{BHLM,DLN}.
	
We have the following expression for connected numbers where the sum over $\vec j \in \mathbb{Z}_+^n$ is finite because of quasi-polynomiality:
\begin{equation*}
h^{\circ,[r]}_{g,\mu} = 
\left(2g - 2 + l(\mu) + \frac{|\mu|}{r}\right)!\cdot {\prod_{i=1}^n\dfrac{\mu_i^{\floor*{\frac{\mu_i}{r}}}}{\floor*{\frac{\mu_i}{r}}!}}\, \sum_{\vec{j} \in \Z_+^n} c_{g,n,\vec{j}}^{ \langle \frac{\vec{\mu}}{r} \rangle}  \mu_1^{j_1}\cdots \mu_{n}^{j_n}.
\end{equation*}
Here $c_{g,n,\vec{j}}^{ \langle \frac{\vec{\mu}}{r} \rangle}$ are the coefficients of the polynomial $P_{g,n}^{\langle \frac{\vec{\mu}}{r} \rangle}(\mu_1,\dots,\mu_n)$ in Theorem \ref{ThPol}.
Hence the partition function reads:
$$
H^{ \circ,[r]}_{g,n}(x_1, \ldots, x_n) = 
\sum_{\substack{\vec{\mu}\\ l(\vec\mu)=n}} \sum_{\vec{j} \in \Z_+^n}
c_{g,n,\vec{j}}^{\langle \frac{\vec{\mu}}{r} \rangle}
\br{\prod_{i=1}^n\dfrac{\mu_i^{\floor*{\frac{\mu_i}{r}} + j_i }}{\floor*{\frac{\mu_i}{r}}!}}
x_1^{\mu_1} \cdots x_n^{\mu_n}
$$
Now we apply the Euclidean division to each $\mu_i$ with the notations:
$$
\mu_i = r \floor*{\frac{\mu_i}{r}} + r \Big{\langle} \frac{\mu_i}{r}\Big{\rangle} \qquad \qquad \sigma_i = \floor*{\frac{\mu_i}{r}}, \qquad
\eta_i = r \Big{\langle} \frac{\mu_i}{r} \Big{\rangle}
$$
The coefficients $c_{g,n,\vec{j}}^{\langle \frac{\vec{\mu}}{r} \rangle}$ only depends on the residue of the $\mu_i$
modulo $r$. Writing $[r-1]$ for $\{0,\dots,r-1\}$ we get:
$$
H^{\circ,[r]}_{g,n}(x_1, \ldots, x_n) = \sum_{\vec{\beta} \in [r-1]^n} \sum_{\vec{j} \in \Z_+^n}
c_{g,n,\vec{j}}^{\beta}
\prod_{i=1}^n \sum_{r\sigma_i+\eta_i > 0} \dfrac{{(r\sigma_i + \eta_i)}^{\sigma_i + j_i }}{\sigma_i !}
x_i^{r\sigma_i + \beta_i}
$$

\begin{lem}
The $n$-point functions $H^{\circ, [r]}_{g,n}(x_1,\ldots,x_n)$ are local expansions around $(x_1,\dots,x_n)=(0,\dots,0)$ of rational functions in $(z_1,\ldots,z_n)$, where 
$$
x(z)=ze^{-z^r}.
$$
\end{lem}

\begin{proof}
It is proved in~\cite[Equation (46)]{SSZ} that
\begin{equation}\label{eq:basic-functions-on-leaves}
\sum_{\sigma=0}^\infty \frac{(r\sigma+{\eta'} )^\sigma}{\sigma!}x^{r\sigma+\eta'} = \frac{z^{\eta'}}{1-rz^r},\quad \eta'=1,\dots,r
\end{equation}
(note that here we use $\eta'=1,\dots,r$ instead of $\eta=0,\dots,r-1$ in order to take uniformly the sum over $\sigma\geq 0$ rather than $r\sigma+\eta>0$).
This is obviously a rational function in $z$, as well as 
\begin{equation}\label{eq:functions-on-leaves}
\sum_{\sigma=0}^\infty \frac{(r\sigma+\eta)^{\sigma+j}}{\sigma!}x^{r\sigma+\eta} = \left(x\frac{d}{dx}\right)^j\frac{z^\eta}{1-rz^r}
=\left(\frac{z}{1-rz^r}\frac{d}{dz}\right)^j\frac{z^\eta}{1-rz^r}.
\end{equation}
So, $H^{\circ, [r]}_{g,n}(x_1,\dots,x_n)$ is an expansion of a finite linear combination of products of rational functions in $z_1,\dots,z_n$.
\end{proof}

Let us denote by $p_1,\dots,p_r$ the critical points of the function $x(z)$. It is obvious that each function $z^{\eta'}/(1-rz^r)$ is a linear combination with constant coefficients of the functions $1/(z-z(p_i))$, $i=1,\dots,r$, up to an additive constant (where said additive constant is not of interest to us, since we are dealing with the differentials of these functions). This implies that all functions given by Equation~\eqref{eq:functions-on-leaves} are linear combinations of
$1/(z-z(p_i))^a$, $i=1,\dots,r$, $a\geq 1$. So, we have:

\begin{lem}
	The symmetric differentials $\omega_{g,n}:=(d_1\otimes\cdots\otimes d_n) H^{\circ, [r]}_{g,n}(x_1,\dots,x_n)$ are equal to the sum of their principal parts in the coordinate $z$ at the points $p_1,\dots,p_n$. 
\end{lem}

This Lemma immediately implies Equation~\eqref{eq:principle-parts} for the standard Cauchy kernel in the coordinate $z$ given by $B(z_1,z_2)=dz_1dz_2/(z_1-z_2)^2$. 

 

\begin{lem} The differentials $\omega_{g,n}(z_1,\dots,z_n)$ satisfy the linear loop equation~\eqref{eq:linear-loop-equation}, namely, $\omega_{g,n}(z_1,\dots,z_n)+\omega_{g,n}(\tau_iz_1,z_2,\dots,z_n)$ is holomorphic for $z_1\to p_i$, where by $\tau_i$ we denote the deck transformation of function $x$ near the point $p_i$, $i=1,\dots,r$.
\end{lem}

\begin{proof} It is sufficient to proof this lemma for the differentials of the functions given by Equation~\eqref{eq:functions-on-leaves}. Observe that the operator $x\frac{d}{dx}$ preserves this property, namely, if $df(z)+df(\tau_iz)$ is holomorphic for $z\to p_i$, then $d(x\frac{d}{dx})f(z)+d(x\frac{d}{dx})f(\tau_iz)$ is also holomorphic for $z\to p_i$. It is proved in~\cite[Equation (4.5)]{SSZ} that 
$$
\frac{z^{\eta'}}{1-rz^r} = \left(x\frac{d}{dx}\right) \frac{z^{\eta'}}{\eta'}, \quad \eta'=1,\dots,r.
$$
The functions $z^{\eta'}$ are holomorphic, so their differentials satisfy the linear loop equation. Therefore, the differentials of all the functions given by Equation~\eqref{eq:functions-on-leaves} satisfy this property as well.
\end{proof}

Now we have to explain how we derive the quadratic loop equation~\eqref{eq:quadratic-loop-equation}. The cut-and-join equation for double Hurwitz numbers~\cite{GJV} (see also \cite{DLN}) can be written in the following form:
\begin{align}
0 &= -(2g-2+n)H^{\circ,[r]}_{g,n}(x_{[n]})-\sum_{i=1}^n (x_i\frac{d}{dx_i})H^{\circ,[r]}_{g,n}(x_{[n]}) \\ \notag
& +\frac 12 \sum_{i\not=j} \left[\frac{x_i}{x_j-x_i} (x_j\frac{d}{dx_j}) 
H^{\circ,[r]}_{g,n-1}(x_{[n]\setminus\{i\}})
+ 
\frac{x_j}{x_i-x_j} (x_i\frac{d}{dx_i}) 
H^{\circ,[r]}_{g,n-1}(x_{[n]\setminus\{j\}})
\right]
\\ \notag 
& +\frac{1}{2}\sum_{i=1}^n \left[(x'\frac{d}{dx'})(x''\frac{d}{dx''})
H^{\circ,[r]}_{g-1,n+1}(x', x'', x_{[n]\setminus \{i\}}) \right]_{\substack{x'=x''=x_i}} 
\\ \notag
& + \frac{1}{2} \sum_{i=1}^n \mathop{\sum_{g_1+g_2=g}}_{I\sqcup J=[n] \setminus \{i\}} \left[ (x_i\frac{d}{d x_i}) H^{\circ,[r]}_{g_1,|I|+1}(x_i,x_I) \right] \left[ (x_i\frac{d}{d x_i})H^{\circ,[r]}_{g_2,|J|+1}(x_i,x_J) \right].
\end{align}
Consider the symmetrization of this expression in variable $x_1$ with respect to the deck transformation near the point $p_i$.  Apply further the operator $\prod_{j=2}^n (\frac{d}{dx_j})$ to it and cancel the terms that do not contribute to the polar part of this expression at $z(x_1)\to p_i$. The obstruction for the derived expression to be holomorphic at $p_i$ is precisely the quadratic loop equation~\eqref{eq:quadratic-loop-equation}. 

This computation implicitly contained in~\cite{DLN} and~\cite{BHLM} as the first step of the derivation of the topological recursion, see also~\cite{DKOSS} for a special case of that. We refer here also to~\cite[Section 2.4]{BS}, where it is shown how to derive the topological recursion from the abstract loop equations in general situation, where one can easily recognize the general pattern of the argument in~\cite{DLN} for this particular case. 
\end{proof}

\section{Johnson-Pandharipande-Tseng formula}
In this section we give a new proof of a special case of the Johnson-Pand\-ha\-ri\-pan\-de-Tseng formula for orbifold Hurwitz numbers. This is a simple corollary of Theorem~\ref{TOPREC}, and the results obtained in~\cite{LPSZ}.

We consider the space $\overline{\mathcal{M}}_{g,-\vec\mu}(B\mathbb{Z}_r)$ of stable maps to the classifying space $B\mathbb{Z}_r$ of the cyclic group of order $r$, where the vector $-\vec\mu$ in the notation corresponds to the prescribing the monodromy data  $ (-\mu_1\mod r,\dots,-\mu_n\mod r)$ at the marked points of source curves. One can think about the elements of this space as admissible covers of curves in $\overline{\mathcal{M}}_{g,n}$ with given monodromy at the marked points. Denote by $p$ the forgetful map $\overline{\mathcal{M}}_{g,-\vec\mu}(B\mathbb{Z}_r)\to \overline{\mathcal{M}}$. 

Consider the action of $\mathbb{Z}_r$ on the $H^0(C,\omega_C)$, where $C$ is the covering curve. Consider its irreducible component that corresponds to the character $U\colon \mathbb{Z}_r\to\mathbb{C}^*$ that send a generator to $\exp(2\pi i/r)$. This component gives us a vector bundle over $\overline{\mathcal{M}}_{g,-\vec\mu}(B\mathbb{Z}_r)$, whose Chern classes we denote by $\lambda_i$, $i\geq 0$. We denote by $S(\vec{\mu})$ the class
\begin{equation}
S(\left<\vec\mu/r\right>):=r^{1-g+\sum \left < \frac{\mu_i}{r} \right >}
p_*\sum_{i\geq 0} (-r)^i \lambda_i.
\end{equation}

\begin{teo} \label{thm:JPT} We have:
\begin{align} \label{eq:particular-jpt}
 \frac{ h^{\circ,[r]}_{g;\vec{\mu}} } {b!} = 
  \prod_{i=1}^{l(\vec{\mu})}\frac{\mu_i^{\floor*{\frac{\mu_i}{r}}}}{\floor{\frac{\mu_i}{r}}!}
  \int_{\overline{\mathcal{M}}_{g,l(\vec{\mu})}}
  \frac{S(\left<\vec\mu/r\right>)}{\prod_{j=1}^{l(\vec{\mu})} (1 - \mu_j {\psi}_j)}
\end{align}
\end{teo}

\begin{rmk}
	This is a special case of the Johnson-Pandharipande-Tseng formula proved in~\cite{JPT}, see also~\cite{J09,DLN,BHLM} for further explanation of the class $S(\vec{\mu})$ used in it. 
\end{rmk}

\begin{rmk}
	Note that this formula looks exactly as formula~\eqref{eq:orbifold-hurwitz-poly}, but now the coefficients of the polynomial $P_{g,n}^{\langle \frac{\vec{\mu}}{r} \rangle}(\mu_1,\dots,\mu_n)$ are
explicitly represented as intersection numbers.
\end{rmk}

We give here a new proof of Theorem~\ref{thm:JPT}.

\begin{proof} The proof consists of two simple observations. On the one hand, Theorem~\ref{TOPREC} says that the expressions 
$$
d_1\otimes\cdots\otimes d_{n}	\sum_{l(\vec{\mu})=n} \frac{ h^{\circ,[r]}_{g;\vec{\mu}} } {b!}x_1^{\mu_1}\cdots x_n^{\mu_n}
$$
are expansions of the symmetric differentials  $\omega_{g,n}(z_1,\dots,z_n)$ that satisfy the topological recursion for the spectral curve data $(\C\mathbb{P}^1, x=ze^{-z^r}, y=z^r)$. On the other hand, it is proved in~\cite{LPSZ} that the expansion of the correlation differentials for this spectral curve is given by 
$$
d_1\otimes\cdots\otimes d_{n} \sum_{l(\vec{\mu})=n}
 \int_{\overline{\mathcal{M}}_{g,n}}
 \frac{S(\left<\vec\mu/r\right>)}{\prod_{j=1}^{n} (1 - \mu_j {\psi}_j)}
  \prod_{i=1}^{n}\frac{\mu_i^{\floor*{\frac{\mu_i}{r}}}}{\floor{\frac{\mu_i}{r}}!}x^{\mu_i}
$$
This identifies the left hand side and the right hand side of Equation~\eqref{eq:particular-jpt}.
\end{proof}

\end{document}